\documentclass{article}
\usepackage{ijcai23}

\usepackage{times}
\usepackage{soul}
\usepackage{url}
\usepackage[hidelinks]{hyperref}
\usepackage[utf8]{inputenc}
\usepackage[small]{caption}
\usepackage{graphicx}
\usepackage{amsmath}
\usepackage{amsthm}
\usepackage{booktabs}
\usepackage{algorithm}
\usepackage{algorithmic}
\usepackage[switch]{lineno}

\author{
Pavel Naumov$^1$
\and
Oliver Orejola$^2$
\affiliations
$^1$University of Southampton, United Kingdom\\
$^2$Tulane University, United States
\emails
p.naumov@soton.ac.uk,
oorejola@tulane.edu
}

\usepackage{amsmath}
\usepackage{amsthm}
\usepackage{booktabs}
\usepackage{algorithm}
\usepackage{algorithmic}
\usepackage{amssymb}

\newtheorem{theorem}{Theorem}
\newtheorem{lemma}{Lemma}
\newtheorem*{claim}{\textnormal{\em Claim}}

\newtheorem{definition}{Definition}

\renewcommand{\phi}{\varphi}
\renewcommand{\epsilon}{\varepsilon}

\newenvironment{proof-of-claim}{\noindent{\em Proof of Claim.}}{\hfill $\boxtimes\hspace{2mm}$\linebreak}

\newcommand{\K}{{\sf K}}
\renewcommand{\H}{{\Box}}
\newcommand{\HH}{{\sf H}}
\renewcommand{\S}{{\sf S}}

\title{Clandestine Strategies in Games with Imperfect Information}

\title{Logic of Clandestine Operations}

\title{Strategic Coalitions in Clandestine Games}

\title{Shhh! The Logic of Clandestine Operations}

\date{February 2020}

\begin{document}

\maketitle

\begin{abstract}
An operation is called covert if it conceals the identity of the actor; it is called clandestine if the very fact that the operation is conducted is concealed. The paper proposes a formal semantics of clandestine operations and introduces a sound and complete logical system that describes the interplay between the distributed knowledge modality and a modality capturing coalition power to conduct clandestine operations.
\end{abstract}

\section{Clandestine Games}

In this paper, we study games in which coalitions can engage in concealed operations. The US Department of Defense Dictionary of Military and Associated Terms distinguishes between covert and clandestine operations. Covert operations are planned and executed to conceal the identity of the actor. An operation is clandestine when the very fact that the operation is conducted is concealed~\cite{20dod}. Thus, every clandestine operation is covert, but not every covert operation is clandestine. The focus of the current work is on clandestine operations. 

\begin{figure}[ht]
\begin{center}
\scalebox{0.45}{\includegraphics{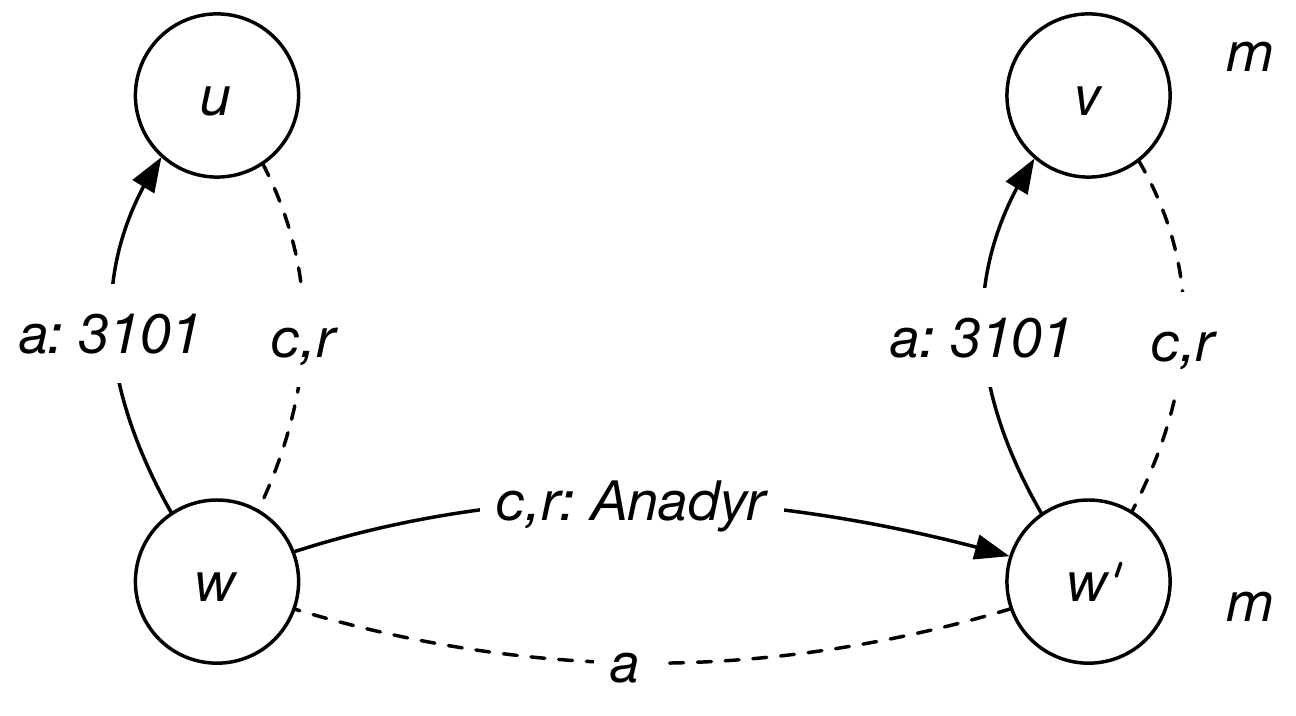}}
\caption{Cuban Missile Crisis Game.}\label{cuban figure}
\end{center}
\end{figure}
An example of a clandestine operation is 1962 Operation Anadyr conducted by the Soviet Union arm forces as a prelude to the Cuban Missile Crisis~\cite{h07cia}. The operation consisted of the delivery and deployment of ballistic missiles with nuclear warheads in Cuba to prevent an invasion of the island by the United States. Figure~\ref{cuban figure} depicts our representation of the Cuban Missile Crisis as a {\em clandestine  game} between three players: the Americans ($a$), the Cubans ($c$), and the Russians ($r$). Operation Anadyr was executed by the Cubans and the Russians and consisted in transitioning the world from state $w$ to state $w'$. Propositional variable $m$ denotes the statement ``Missiles are deployed in Cuba''. It is false in state $w$ and true in state $w'$. Operation Anadyr was {\em concealed} in the sense that the Americans were not able to detect the transition of the world from state $w$ to state $w'$. In the diagram, the indistinguishability of these two states to Americans is shown using a dashed line. 

Although states $w$ and $w'$ are indistinguishable to Americans, this does not prevent them from discovering the transition from state $w$ to state $w'$ by executing an operation of their own. In fact, they did just that on October 14th, 1962, by conducting a clandestine operation Mission 3101~\cite{m92cia}. Mission 3101 consisted of a U-2 spy plane secretly flying over Cuban territory to collect military intelligence. Mission 3101 also was concealed in the sense that, as shown in the diagram, the Cubans and the Russians were not able to detect its execution that transitioned the world from state $w'$ to state $v$. If the same Mission 3101 were to be executed in state $w$, it would hypothetically transition the world from state $w$ to state $u$. The Americans can distinguish state $v$ from state $u$ based on the reconnaissance photos taken by the spy plane. This explains how the Americans were able to detect the execution of Operation Anadyr through operation Mission 3101.

Coalition power in games with imperfect information has been studied before in {\em synchronous} settings where all agents act at once and, thus, everyone is aware that something happened~\cite{vw03sl,aa12aamas,nt17aamas,nt18ai,nt18aaai,aa16jlc}. To capture clandestine operations it is crucial to use semantics in which an agent might be unaware of the game transitioning from one state to another as a result of the actions of other agents. Such a behaviour could be modelled, for example, by extending the semantics of the above logical systems with a single $sleep$ action. Additionally, it should be required that any agent executing action $sleep$ should not be able to distinguish the initial and the final state of any transition during which the agent used $sleep$. This approach would also need to settle who learns what if two or more disjoint coalitions execute clandestine operations synchronously.

For the sake of the clarity of presentation, in this paper, we define the semantics of clandestine operations in terms of a class of asynchronous games that we call {\em clandestine games} that are described in the definition below. 

In this paper, we will assume a fixed set of agents $\mathcal{A}$. By a coalition we mean any (possibly empty) subset of $\mathcal{A}$. For any coalition $C$, by $\overline{C}$ we denote the complement of set $C$ with respect to set $\mathcal{A}$.

\begin{definition}\label{game}
Let a {\bf\em clandestine game} be any such tuple $(W,\{\sim_a\}_{a\in \mathcal{A}},\Delta,M,\pi)$ that
\begin{enumerate}
    \item $W$ is a set of ``states''.
    \item $\sim_a$ is an ``indistinguishability'' equivalence relation on set $W$ for each agent $a\in \mathcal{A}$. We write $w\sim_C u$ if $w\sim_a u$ for each agent $a\in C$.
    \item $\Delta$ is a nonempty set of ``operations''.
    \item $M$ is a set of tuples $(w,C,\delta,u)$, where $w,u\in W$ are states, $C\subseteq\mathcal{A}$ is a coalition, and $\delta\in\Delta$ is an operation. It is assumed that set $M$, called ``mechanism'', satisfies the following two conditions
    \begin{enumerate}
    \item {\bf\em concealment}: for any two states $w,u\in W$, any coalition of agents $C\subseteq\mathcal{A}$, any operation $\delta\in\Delta$, if $(w,C,\delta,u)\in M$, then $w\sim_{\overline{C}} u$,
    \item {\bf\em nontermination}: for any state $w\in W$, any coalition of agents $C\subseteq\mathcal{A}$, and any operation $\delta\in\Delta$, there is at least one state $u\in W$ such that $(w,C,\delta,u)\in M$.
    \end{enumerate}
    \item $\pi(p)$ is a subset of $W$ for each propositional variable $p$.
\end{enumerate}
\end{definition}
The diagram in Figure~\ref{cuban figure} depicts an example of a clandestine game with four states ($w$, $w'$, $u$, and $v$) and two operations (Anadyr and 3101). The indistinguishability relations are shown by dashed lines and the mechanism is depicted by directed lines. The diagram omits loop operations. This means, for example, that if {\em the Americans} execute Operation Anadyr in any of the states, then the game transitions back to the same state. The nontermination condition 4(b) guarantees that no operation can terminate a game without reaching some state.

In a real-world setting, a variety of operations might be performed by any coalition. Some of them satisfy the concealment condition 4(a) of Definition~\ref{game}, the others might not. We excluded non-concealed operations from our games to keep the presentation simple. If such operations are added to the models and the quantifier over operations $\delta$ in item 5 of Definition~\ref{sat} below is simultaneously restricted to concealed operations only, then the soundness and the completeness results of this paper will remain true and no changes to their proofs will be necessary.

In this paper, we propose a sound and complete logical system for reasoning about coalition power to conduct clandestine operations. The rest of the paper is organized as follows. In the next section, we discuss the interplay between knowledge and actions and explain why existing coalition power modalities do not capture the properties of clandestine operations. Then, we define the syntax and semantics of our logical system. In the section Coalition-Informant-Adversary Principle, we introduce and discuss the most non-trivial axiom of our system.  In the section that follows, we list the remainder of the axioms. After that, we sketch the completeness of our logical system. The proof of soundness and some details of the completeness are in the appendix.

\section{Knowledge and Actions}

In this section, we discuss how different forms of knowledge can be captured in the existing modal logics for reasoning about coalition power and explain why the power to perform a clandestine operation is not expressible in these logics.

When discussing the interplay between knowledge and actions, it is common to distinguish {\em ex-ante}, {\em interim}, and {\em ex-post} knowledge. They refer to an agent's (or a coalition's) knowledge before the action, at the moment of the action, and after the action, respectively. One of the first logical systems describing the interplay between distributed knowledge modality $\K_C$ and coalition power modality $\S_C$ was introduced in~\cite{aa12aamas}. Using their language, one can write $\K_C\S_C\phi$ to state that coalition $C$ knows {\em ex-ante} (before the action) that it has a strategy (joint action) to achieve $\phi$. Using the same language, one can write $\S_C\K_C\phi$ to state that coalition $C$ has a strategy that would result in $\phi$ being known {\em ex-post} to the coalition. The language of~\cite{aa12aamas} cannot be used to express {\em interim} knowledge. However, this could be done using ``seeing to it'' modality~\cite{bp90krdr,h01,h95jpl,hp17rsl,ow16sl}.

Knowing that a strategy exists, as in $\K_C\S_C\phi$, is different from actually knowing the strategy. If a coalition $C$ knows {\em ex-ante} what strategy it can use to achieve $\phi$, then we say that the coalition has a {\em know-how} strategy to achieve $\phi$ and denote this by $\HH_C\phi$. Unless the coalition has a perfect recall, knowing {\em ex-ante} a strategy to achieve $\phi$ does not imply knowing ex-ante a strategy that results in knowing {ex-post} that $\phi$ is achieved. The latter, however, could be expressed as $\HH_C\K_C\phi$. The interplay between coalitional know-how modality $\HH_C$ and distributed knowledge modality $\K_C$ has been studied in~\cite{nt17aamas,fhlw17ijcai,nt18ai,nt18aaai,aa16jlc,cn20ai}. 

\begin{figure}[ht]
\begin{center}
\scalebox{0.45}{\includegraphics{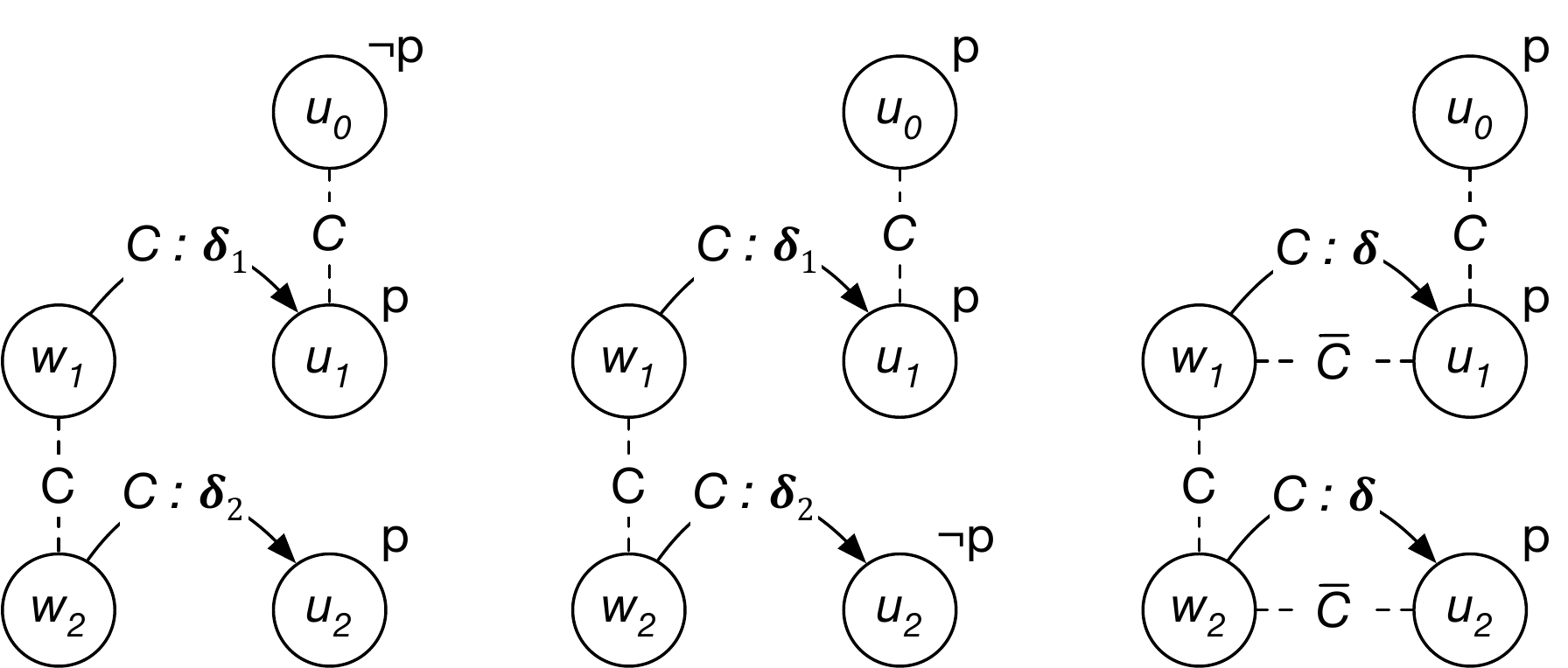}}
\caption{Knowledge and Actions.}\label{knowledge action figure}
\end{center}
\end{figure}
In epistemic models, knowledge is usually captured by an indistinguishably relation between states. For example, in Figure~\ref{knowledge action figure} (left), $w_1\Vdash \K_C\S_C p$. In other words, coalition $C$ knows ex-ante that it has a strategy to achieve $p$. This is true because the coalition has such a strategy not only in state $w_1$, but also in state $w_2$, indistinguishable to the coalition from $w_1$. Note that this is not a know-how strategy because the required strategy in state $w_1$ (strategy $\delta_1$) is different from the required  strategy in state $w_2$ (strategy $\delta_2$). Thus, $w_1\Vdash \neg\HH_C p$. Note also that in state $w_1$ coalition $C$ does not have a strategy to achieve ex-post knowledge of $p$. We write this as  $w_1\Vdash \neg\S_C\K_C p$. This is true because state $u_1$ is indistinguishable from state $u_0$ where $p$ is not satisfied.

The situation is different in Figure~\ref{knowledge action figure} (centre). Here, coalition $C$ has a strategy in state $w_1$ to achieve $p$, but the coalition does not know this ex-ante because it cannot distinguish state $w_1$ from state $w_2$ where such a strategy does not exist. Using our notations, $w_1\Vdash \S_C p$ and $w_1\Vdash \neg\K_C\S_C p$. Note, however, that in this setting coalition also has a strategy to achieve ex-post knowledge of $p$ because $p$ is satisfied not only in state $u_1$ but also in state $u_0$, indistinguishable to $C$ from state $u_1$. We write this as $w_1\Vdash \S_C\K_C p$

The clandestine operations that we consider in this paper are know-how strategies. Furthermore, for the reason we discuss in the next section, they are know-how strategies to achieve ex-post knowledge. This alone would not require a new modality because it can be captured in existing know-how logics as $\HH_C\K_C\phi$. However, the last formula does not account for the concealed  nature of clandestine operations. We capture the late by requiring the initial and the final state of the operation to be indistinguishable to the {\em complement} $\overline{C}$ of coalition $C$. Strategy $\delta$ depicted in Figure~\ref{knowledge action figure} (right) is a clandestine operation of coalition $C$ to achieve $p$. In this paper, we introduce a new modality $\Box_C\phi$ to denote an existence of a clandestine operation of coalition $C$ to achieve $\phi$. This modality is not definable through existing modalities of coalition power, know-how, and seeing-to-it, because these existing modalities cannot capture the indistinguishably (by the complement of coalition $C$) of the initial and the final state of the operation.

\section{Syntax and Semantics}

Language $\Phi$ of our logical system is defined by the grammar

$$
\phi := p\;|\;\neg\phi\;|\;\phi\to\phi\;|\;\K_C\phi\;|\;\H_C\phi, 
$$
where $p$ is a propositional variable and $C$ is a coalition. We read formula $\K_C\phi$ as ``coalition $C$ knows $\phi$'', and formula $\Box_C\phi$ as ``coalition $C$ knows which clandestine operation it can execute to achieve $\phi$''. In both cases, the knowledge is distributed. We assume that Boolean constants $\top$ and $\bot$ as well as disjunction $\vee$ are defined in the standard way. We use $\K_{C,D}\phi$ and $\Box_{C,D}\phi$ as shorthand for $\K_{C\cup D}\phi$ and $\Box_{C\cup D}\phi$ respectively.

In the definition below, item~5  gives formal semantics of modality $\H_C\phi$, see Figure~\ref{satisfiability figure}. 

\begin{definition}\label{sat}
For any state $w\in W$ of a clandestine game $(W,\{\sim_a\}_{a\in \mathcal{A}},\Delta,M,\pi)$ and any formula $\phi\in\Phi$, satisfiability relation $w\Vdash\phi$ is defined recursively as
\begin{enumerate} 
    \item $w\Vdash p$ if $w\in \pi(p)$,
    \item $w\Vdash \neg\phi$ if $w\nVdash \phi$,
    \item $w\Vdash\phi\to\psi$ if $w\nVdash\phi$ or $w\Vdash\psi$,
    \item $w\Vdash\K_C\phi$ if $u\Vdash\phi$ for any $u\in W$ such that $w\sim_C u$,
    \item $w\Vdash \H_C\phi$ if there is a nonempty coalition $C'\subseteq C$ and  an operation $\delta\in \Delta$ such that for any states $w',u,u'\in W$, if $w\sim_C w'$, $(w',C',\delta,u)\in M$, and $u\sim_C u'$, then $u'\Vdash \phi$.
\end{enumerate}
\end{definition}

\begin{figure}[ht]
\begin{center}
\scalebox{0.45}{\includegraphics{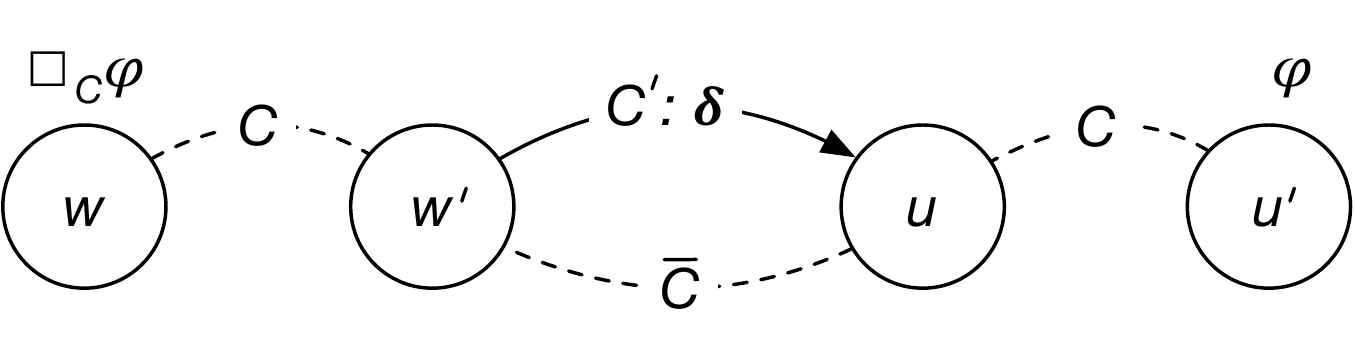}}
\caption{Towards item 5 of Definition~\ref{sat}.}\label{satisfiability figure}
\end{center}
\end{figure}
In item 5 of the above definition, we introduce coalition $C'$ to capture the fact that in order for a coalition $C$ to know a clandestine operation to achieve a certain goal, not all members of the coalition $C$ have to take an active part in it. 

Recall that Definition~\ref{game} allows for some operations to be conducted by the empty coalition. Such operations can change the state of the game. However, according to the concealment condition of Definition~\ref{game}, such change is not noticeable to any agent in the game. Informally, these operations could be thought of as nondeterministic transitions of the game that occur independently from the actions of the agents and are not noticeable to them. The presence of such transitions is not significant for the results in this paper. We do not exclude them for the sake of generality. At the same time, in Definition~\ref{sat}, we require coalition $C'$ to be nonempty. Intuitively, a coalition can ask some of its members to conduct an operation, but the coalition cannot ask the empty coalition, because operations of the empty coalition are system transitions not controlled by the agents. The restriction of $C'$ to nonempty coalitions is significant for our results. 

Item 5 of Definition~\ref{sat} is using state $w'$ to express that the clandestine operation $\delta$ not only exists but it is known to coalition $C$. Note that this knowledge, captured through statement $w\sim_C w'$, is the knowledge of the whole coalition $C$, not just its part $C'$ that executes the operation. In other words, we assume that some members of the coalition $C$ could be passive {\em informants}. We explore this in the Coalition-Informant-Adversary axiom of our logical system. 

Formula $\H_C\phi$ states that coalition $C$ knows a clandestine operation to achieve $\phi$. Because clandestine games are asynchronous, an important question is for how long $\phi$ will remain true after the operation. If another coalition can ``undo'' the operation without $C$ even noticing, then coalition $C$ could only be sure that $\phi$ holds at the very moment the operation is completed. To avoid this, in item 5 of Definition~\ref{sat}, we require $\phi$ to be satisfied not only in the completion state $u$ of the operation $\delta$, but also in all states $u'$ indistinguishable from state $u$ by coalition $C$. In other words, statement $\phi$ remains true until at least one of the members of coalition $C$ takes part in another clandestine operation\footnote{If non-concealed operations are added to Definition~\ref{game} as described in the previous section, then $\phi$ will remain until at least one of the members of coalition $C$ becomes aware that another operation took place.}.

\section{Coalition-Informant-Adversary Principle}

The most interesting axiom of our logical system is a principle that captures strategic information dynamics between three sets of agents: a {\em coalition} that conducts a clandestine operation, a group of {\em informants} who passively cooperate with the coalition by sharing knowledge but do not participate in the operation itself, and a group of {\em adversaries} who do not cooperate with the coalition at all. To understand this principle, let us first consider its simplified form without the adversaries: for any disjoint coalitions $C$ and $I$,
\begin{equation}\label{baby CIA principle}
    \K_{I}(\K_C\phi\to\K_{C}\psi)\to(\H_C\phi\to\H_{C,I}\psi).
\end{equation}
The assumption $\H_C\phi$ of this principle states that before the operation ({\em ex-ante}) the coalition knows which clandestine operation it should conduct in order to know after the operation ({\em ex-post}) that $\phi$ is true. The other assumption $\K_{I}(\K_C\phi\to\K_{C}\psi)$ of this principle refers to {\em ex-ante} knowledge of a group of informants $I$. Because the operation is clandestine and $C\cap I=\varnothing$, the  {\em ex-ante} and {\em ex-post} knowledge of $I$ is the same. Thus, statement $\K_C\phi\to\K_{C}\psi$ will have to be true after the operation. In other words, after the operation coalition, $C$ will know that not only condition $\phi$, but also condition $\psi$ is true. Coalition $C$ alone, however, does not know this {\em ex-ante} and thus, it alone does not know an operation to achieve condition $\psi$. Nevertheless, recall that coalition $I$ knows $\K_C\phi\to\K_{C}\psi$ {\em ex ante}. Thus, it knows {\em ex-ante} that $\K_C\phi\to\K_{C}\psi$ will have to be true after any clandestine operation that does not involve $I$. Therefore, the union of the coalitions $C$ and $I$ knows {\em ex-ante} the operation that $C$ can conduct to achieve $\psi$. That is, $\H_{C,I}\psi$. 

Note that the purpose of modality $\K_I$ in the assumption $\K_{I}(\K_C\phi\to\K_{C}\psi)$ of principle~(\ref{baby CIA principle}) is to make sure that statement $\K_C\phi\to\K_{C}\psi$ is preserved during the clandestine operation of coalition $C$. If one were to consider an additional coalition, that we refer to as an adversary coalition $A$, then replacing modality $\K_I$ with $\K_{A,I}$ still guarantees that statement $\K_C\phi\to\K_{C}\psi$ is preserved during the operation (as long as $A$ is also disjoint with $C$). Thus, one might think that the following form of principle~(\ref{baby CIA principle}) is also valid:
$$
\K_{A,I}(\K_C\phi\to\K_{C}\psi)\to(\H_C\phi\to\H_{C,I}\psi).
$$
This statement, however, is not true. Assumptions $\K_{A,I}(\K_C\phi\to\K_{C}\psi)$ and $\H_C\phi$ can only guarantee that coalition $C$ knows {\em ex ante} an operation to achieve $\phi$. If this operation is executed, then coalition $C\cup I$ will know {\em ex-post} that $\phi$ is true, but {\em they might not know ex-ante that they will know this ex-post}. To make sure that they indeed have such {\em ex-ante} knowledge, one more knowledge modality should be added to the formula:
$$
\K_{C,I}\K_{A,I}(\K_C\phi\to\K_{C}\psi)\to(\H_C\phi\to\H_{C,I}\psi).
$$
Finally, note that if instead of preserving $\K_C\phi\to\K_{C}\psi$, it is enough just to be able to preserve  statement $\K_C\phi\to\K_{C,I}\psi$:
$$
\K_{C,I}\K_{A,I}(\K_C\phi\to\K_{C,I}\psi)\to(\H_C\phi\to\H_{C,I}\psi).
$$
As it turns out, the above formula is the final and the most general form of the Coalition-Informant-Adversary principle. In this paper, we show that this principle, in combination with several much more straightforward other axioms, forms a logical system that can derive all universally valid properties of clandestine operations.


















\section{Axioms}

In addition to propositional tautologies in the language $\Phi$, our logical system contains the following axioms:

\begin{enumerate}
    \item Truth: $\K_C\phi\to\phi$,
    \item Negative Introspection: $\neg\K_C\phi\to\K_C\neg\K_C\phi$,
    \item Distributivity: $\K_C(\phi\to\psi)\to(\K_C\phi\to\K_C\psi)$,
    \item Monotonicity: $\K_{C'}\phi\to\K_C\phi$, where $C'\subseteq C$,
    \item Strategic Introspection: $\H_C\phi\to\K_C\H_C\phi$,
   \item Coalition-Informant-Adversary: if $C\cap (I\cup A)=\varnothing$, then\\ $\K_{C,I}\K_{A,I}(\K_C\phi\to\K_{C,I}\psi)\to(\H_C\phi\to\H_{C,I}\psi)$,
   \item Nontermination: $\neg\H_C\bot$,
   \item Empty Coalition: $\neg\H_\varnothing\phi$.
\end{enumerate}

We write $\vdash\phi$ if a formula $\phi$ is provable from the above axioms using the Modus Ponens and the two Necessitation inference rules:
$$
\dfrac{\phi,\phi\to\psi}{\psi}
\hspace{8mm}
\dfrac{\phi}{\K_C\phi}
\hspace{8mm}
\dfrac{\phi,\;\;\;\;\; C\neq\varnothing}{\H_C\phi}.
$$
We write $X\vdash\phi$ if the formula $\phi$ is provable from the theorems of our logical system and the set of additional axioms $X$ using only the Modus Ponens inference rule.
The next two lemmas state well-known properties of S5 modality  $\K$. 

The next five lemmas are used in the proof of the completeness of our logical system. We give their proofs in the appendix.

\begin{lemma}\label{super distributivity}
If $\phi_1,...,\phi_n\!\vdash\!\psi$, then $\K_C\phi_1,...,\K_C\phi_n\!\vdash\!\K_C\psi$.
\end{lemma}

\begin{lemma}\label{positive introspection lemma}
$\vdash \K_C\phi\to\K_C\K_C\phi$. 
\end{lemma}

\begin{lemma}\label{KKK lemma}
$\vdash \K_F\K_{E}\K_F\phi\to \H_F\phi$, where $F\nsubseteq E$.
\end{lemma}

\begin{lemma}\label{K neg K to neg H lemma}
$\vdash \K_E\neg\K_F\phi\to \neg\H_F\phi$, where $E\cap F=\varnothing$.
\end{lemma}

\begin{lemma}\label{K over vee lemma}
$\vdash\K_F(\K_E\phi\vee\psi)\to\K_E\phi\vee \K_F\psi$, where $E\subseteq F$.
\end{lemma}

\noindent We show the following soundness theorem in the appendix.

\begin{theorem}\label{completeness}
For any state $w$ of a clandestine game, if $w\Vdash\chi$ for each formula $\chi\in X$ and $X\vdash\phi$, then $w\Vdash\phi$.
\end{theorem}

In the rest of this paper, we prove the completeness of our logical system.

\section{Canonical Model}

As usual, the key step in the proof of completeness is the construction of a canonical model.

The standard canonical model for epistemic logic of individual knowledge S5 defines states as maximal consistent sets of formulae. Two such sets are indistinguishable to an agent $a$ if they contain the same $\K_a$-formulae. Unfortunately, this approach does not work for distributed knowledge because any two sets that are indistinguishable to agents $a$ and $b$ would then only share $\K_a$ and $\K_b$ formulae. They might have different $\K_{a,b}$-formulae. 

\begin{figure}[ht]
\begin{center}
\scalebox{0.45}{\includegraphics{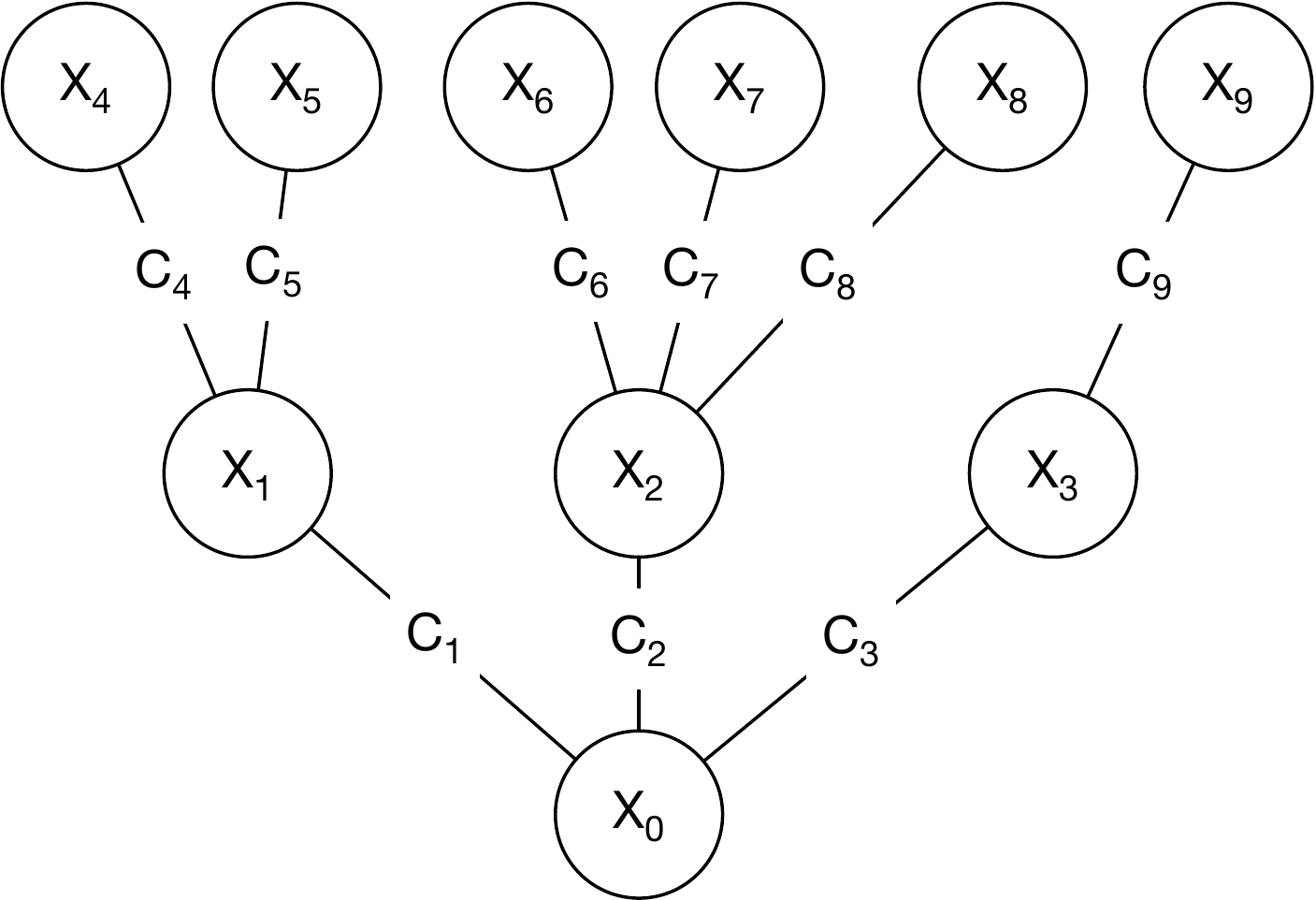}}
\caption{Tree Construction.}\label{tree figure}
\end{center}
\end{figure}

To address this issue, we define the canonical model using a tree whose nodes are labelled by maximal consistent sets and whose edges are labelled by coalitions, see Figure~\ref{tree figure}. 
Informally, states are the nodes of this tree. Formally, any state is a sequence of labels along the path leading from the root of the tree to a node of the tree. 

Let us now define canonical clandestine game $M(X_0)=(W,\{\sim_a\}_{a\in \mathcal{A}},\Delta,M,\pi)$ for an arbitrary maximal consistent set of formulae $X_0$.

\begin{definition}\label{canonical state}
Set $W$ consists of all finite sequences $X_0,C_1,\dots,C_n,X_n$, such that $n\ge 0$ and
\begin{enumerate}
    \item $X_i$ is a maximal consistent set of formulae for all $i>1$,
    \item $C_i\subseteq\mathcal{A}$ is a coalition for all $i\le n$,
    \item $\{\phi\in\Phi\;|\;\K_{C_{i}}\phi\in X_{i-1}\}\subseteq X_{i}$, for all $i\le n$.
\end{enumerate}
\end{definition}

We define a tree structure on the set of states $W$ by saying that state $w=X_0,C_1,X_1,C_2,\dots,C_n,X_n$ and state $w::C_{n+1}::X_{n+1}$ are connected by an undirected edge labeled with all agents in coalition $C_{n+1}$. For example, for the tree depicted in Figure~\ref{tree figure}, state $X_0,C_2,X_2$ is adjacent to state $X_0,C_2,X_2,C_8,X_8$ and the edge between them is labelled with all agents in coalition $C_8$.

\begin{definition}\label{canonical sim}
For any two states $w,w'\in W$ and any agent $a\in\mathcal{A}$, let $w\sim_a w'$ if all edges along the simple path between $w$ and $w'$ are labelled with agent $a$.
\end{definition}
Note that, in the above definition, the path might consist of a single node. 
\begin{lemma}\label{canonical sim is equivalence relation}
Relation $\sim_a$ is an equivalence relation on set $W$.
\end{lemma}

\begin{definition}
Set of operations $\Delta$ is the set of all formulae in language $\Phi$.
\end{definition}

Informally, operation $\phi\in\Delta$ is a clandestine operation in the canonical game that achieves $\phi$ unnoticeable to the agents outside  of the coalition that performed the operation and makes the result known to the coalition. This intuition is captured in the definition below. Throughout the paper, by $hd(w)$ we denote the last element of the sequence $w$.

\begin{definition}\label{canonical mechanism}
Canonical mechanism $M$  is a set of all tuples $(w,C,\phi,u)$ where $w,u\in W$ are states, $C\subseteq\mathcal{A}$ is a coalition, and $\phi\in\Phi$ is a formula, such that $w\sim_{\overline{C}} u$ and  if $\H_C\phi\in hd(w)$, then $\K_C\phi\in hd(u)$.
\end{definition}

Note that the requirement $w\sim_{\overline{C}} u$ in the above definition implies that mechanism $M$ satisfies the concealment condition from Definition~\ref{game}. Next, we show that $M$ also satisfies the nontermination condition.

\begin{lemma}
For any state $w$, any coalition $C\subseteq\mathcal{A}$, and any formula $\phi\in\Phi$, there is a state $u\in W$ such that  $(w,C,\phi,u)\in M$.
\end{lemma}
\begin{proof} We consider the following two cases separately:

\vspace{1mm}
\noindent{\bf Case I:} $\H_C\phi\in hd(w)$. Let
$$
X=\{\K_C\phi\}\cup \{\psi\;|\;\K_{\overline{C}}\psi\in hd(w)\}.
$$
\begin{claim}
Set $X$ is consistent.
\end{claim}
\begin{proof-of-claim}
Assume the opposite. Thus, 
there are formulae $\K_{\overline{C}}\psi_1$,\dots, $\K_{\overline{C}}\psi_n\in hd(w)$
such that
$
\psi_1,\dots,\psi_n\vdash \neg\K_C\phi.
$
Hence, 
$
\K_{\overline{C}}\psi_1,\dots,\K_{\overline{C}}\psi_n\vdash \K_{\overline{C}}\neg\K_C\phi.
$
by Lemma~\ref{super distributivity}.
Then, 
$
hd(w)\vdash\K_{\overline{C}}\neg\K_C\phi
$
by the assumption 
$\K_{\overline{C}}\psi_1,\dots, \K_{\overline{C}}\psi_n\in hd(w)$.
Thus,
$
hd(w)\vdash \neg\H_{C}\phi
$
by Lemma~\ref{K neg K to neg H lemma} and the Modus Ponens inference rule.
Then, $\H_{C}\phi\notin hd(w)$ because set $hd(w)$ is consistent, which contradicts the assumption of the case. 
\end{proof-of-claim}

Let $X'$ be any maximal consistent extension of set $X$ and $u$ be the sequence $w::\overline{C}::X'$. Then, $w\in W$ by Definition~\ref{canonical state} as well as the choice of sets $X$ and $X'$.

Finally, note that $w\sim_{\overline{C}} u$ by Definition~\ref{canonical sim} because $u=w::\overline{C}::X'$. Also, $\K_C\phi\in X\subseteq X'=hd(u)$ by the choice of sets $X$ and $X'$ and the choice of sequence $u$. Therefore,  $(w,C,\phi,u)\in M$ by Definition~\ref{canonical mechanism}.

\vspace{1mm}
\noindent{\bf Case II:}  $\H_C\phi\notin hd(w)$. Take $u$ to be world $w$.
%
%
%
%
Therefore,  $(w,C,\phi,u)\in M$ by Definition~\ref{canonical mechanism}. 
This concludes the proof of the lemma.
\end{proof}

\begin{definition}\label{canonical pi}
$\pi(p)=\{w\in W\;|\; p\in hd(w)\}$.
\end{definition}

This concludes the definition of the canonical model $M(X_0)=(W,\{\sim_a\}_{a\in \mathcal{A}},\Delta,M,\pi)$.



\section{Completeness}

As usual, the proof of completeness is using an ``induction'' (or ``truth'') lemma to connect the syntax of our system with the semantics of the canonical model. In our case, this is Lemma~\ref{induction lemma}. The next five lemmas are auxiliary statements that will be used in different cases of the induction step in the proof of Lemma~\ref{induction lemma}. 

\begin{lemma}\label{edge transport lemma}
$\K_D\phi\in X_n$ iff $\K_D\phi\in X_{n+1}$ for any formula $\phi\in\Phi$, any $n\ge 0$, and any state $X_0,C_1,X_1,C_2,\dots,X_n,C_{n+1},X_{n+1}\in W$, and any coalition $D\subseteq C_{n+1}$.
\end{lemma}
\begin{proof}
If $\K_D\phi\in X_n$, then $X_n\vdash\K_D\K_D\phi$ by Lemma~\ref{positive introspection lemma}. Hence, $X_n\vdash\K_{C_{n+1}}\K_D\phi$ by the Monotonicity axiom, the assumption $D\subseteq {C_{n+1}}$, and the Modus Ponens  inference rule. Thus, $\K_{C_{n+1}}\K_D\phi\in X_n$ by the maximality of set $X_n$. Therefore, $\K_D\phi\in X_{n+1}$ by Definition~\ref{canonical state}.

Suppose that $\K_D\phi\notin X_n$. Hence, $\neg\K_D\phi\in X_n$ by the maximality of set $X_n$. Thus, $X_n\vdash \K_D\neg\K_D\phi$ by the Negative Introspection axiom and the Modus Ponens  inference rule. Hence,  $X_n\vdash \K_{C_{n+1}}\neg\K_D\phi$ by the Monotonicity axiom, the assumption $D\subseteq C_{n+1}$, and the Modus Ponens  inference rule. Then, $\K_{C_{n+1}}\neg\K_D\phi\in X_n$ by the maximality of set $X_n$. Thus, $\neg\K_D\phi\in X_{n+1}$ by Definition~\ref{canonical state}. Therefore, $\K_D\phi\notin X_{n+1}$ because set $X_{n+1}$ is consistent. 
\end{proof}
\begin{lemma}\label{K all lemma}
If $\K_C\phi\in hd(w)$ and $w\sim_C u$, then $\phi\in hd(u)$.
\end{lemma}
\begin{proof}
Assumption $w\sim_C u$ implies that all edges along the unique simple path between nodes $w$ and $u$ are labeled with all agents in coalition $C$. Thus, $\K_C\phi\in hd(u)$ by Lemma~\ref{edge transport lemma}. Hence, $hd(u)\vdash \phi$ by the Truth axiom and the Modus Ponens inference rule. Therefore, $\phi\in hd(u)$ because set $hd(u)$ is maximal.
\end{proof}
\begin{lemma}\label{K exists lemma}
If $\K_C\phi\notin hd(w)$, then there is $u\in W$ such that $w\sim_C u$ and $\phi\notin hd(u)$.
\end{lemma}
\begin{proof}
Consider set $X=\{\neg\phi\}\cup\{\psi\;|\;\K_C\psi\in hd(w)\}$.
\begin{claim}
Set $X$ is consistent.
\end{claim}
\begin{proof-of-claim}
Suppose the opposite. Thus, there are 
formulae 
$\K_C\psi_1,\dots,\K_C\psi_n\in hd(w)$ such that 
$\psi_1,\dots,\psi_n\vdash \phi$.
Hence, $\K_C\psi_1,\dots,\K_C\psi_n\vdash \K_C\phi$ by Lemma~\ref{super distributivity}. Then, 
$hd(w)\vdash  \K_C\phi$ by the assumption $\K_C\psi_1,\dots,\K_C\psi_n\in hd(w)$. Thus, $\K_C\phi\in hd(w)$ because set $hd(w)$ is maximal, which contradicts the assumption of the lemma.
\end{proof-of-claim}
Let $X'$ be any maximal consistent extension of set $X$ and $u$ be the sequence $w::C::X'$. Then, $w\in W$ by Definition~\ref{canonical state} as well as the choice of sets $X$ and $X'$. 

Finally, $\neg\phi\in X\subseteq X'=hd(u)$ by the choice of sets $X$ and $X'$ and the choice of sequence $u$. Therefore, $\phi\notin hd(u)$ because set $hd(u)$ is consistent.
\end{proof}

\begin{lemma}\label{H all lemma}
For any formula $\H_C\phi\in hd(w)$ and any three states $w',u,u'\in W$, if $w\sim_C w'$, $(w',C,\phi,u)\in M$, and $u\sim_C u'$, then $\phi\in hd(u')$.
\end{lemma}
\begin{proof}
Assumption $\H_C\phi\in hd(w)$ implies that $hd(w)\vdash \K_C\H_C\phi$ by the Strategic Introspection axiom and the Modus Ponens inference rule. Hence, $\K_C\H_C\phi\in hd(w)$ because set $hd(w)$ is maximal. Thus, $\H_C\phi\in hd(w')$ by Lemma~\ref{K all lemma} and the assumption $w\sim_C w'$. Then, $\K_C\phi\in hd(u)$ by Definition~\ref{canonical mechanism} and the assumption  $(w',C,\phi,u)\in M$. Therefore,  $\phi\in hd(u')$ by Lemma~\ref{K all lemma} and the assumption  $u\sim_C u'$.
\end{proof}
\begin{lemma}\label{H exists lemma}
If $\H_F\phi\notin hd(w)$, then for any nonempty coalition $E\subseteq F$ and any action $\delta\in \Delta$, there  are states $w',u,u'$ such that $w\sim_F w'$, $(w',E,\delta,u)\in M$, $u\sim_F u'$, and $\phi\notin hd(u')$.
\end{lemma}

In the proof of this lemma located below, we consecutively construct states $w'$, $u$, and $u'$. To guarantee that state $u'$ could be constructed after state $u$, we construct a state $u$ such that set $hd(u)$ contains formula $\neg\K_F\phi$. In this case, by Lemma~\ref{K exists lemma}, there must exist a state $u'$ such that $u\sim_F u'$, and $\phi\notin hd(u')$.

One might think that state $u$ could be constructed from state $w'$ in a similar fashion by guaranteeing first that set $hd(w')$ contains formula $\neg\K_{\overline{E}}\K_F\phi$. 
However, there is a problem because Definition~\ref{canonical mechanism} states that if set $hd(w')$ contains formula $\Box_E\delta$, then set $hd(u)$, in addition to formula $\neg\K_F\phi$, must also contain formula $\K_E\delta$. Thus, there are two possible ways sets $hd(w')$ and $hd(u)$ could be constructed:
\begin{enumerate}
    \item[I.] Set $hd(u)$ contains $\neg\K_F\phi$ and  $\K_E\delta$. In this case, set $hd(w')$ must contain formula $\neg\K_{\overline{E}}\,\neg(\neg\K_F\phi\wedge\K_E\delta)$. The last formula is equivalent to $\neg\K_{\overline{E}} (\K_{E}\delta\to\K_F\phi )$,
    \item[II.] Set $hd(u)$ contains only formula $\neg\K_F\phi$. In this case, set $hd(w')$ must contain formulae $\neg\K_{\overline{E}}\K_F\phi$ and $\neg\Box_E\delta$. 
\end{enumerate}
We visualise these two cases on the diagram in Figure~\ref{pink-blue figure}.

\begin{figure}[ht]
\begin{center}
\scalebox{0.45}{\includegraphics{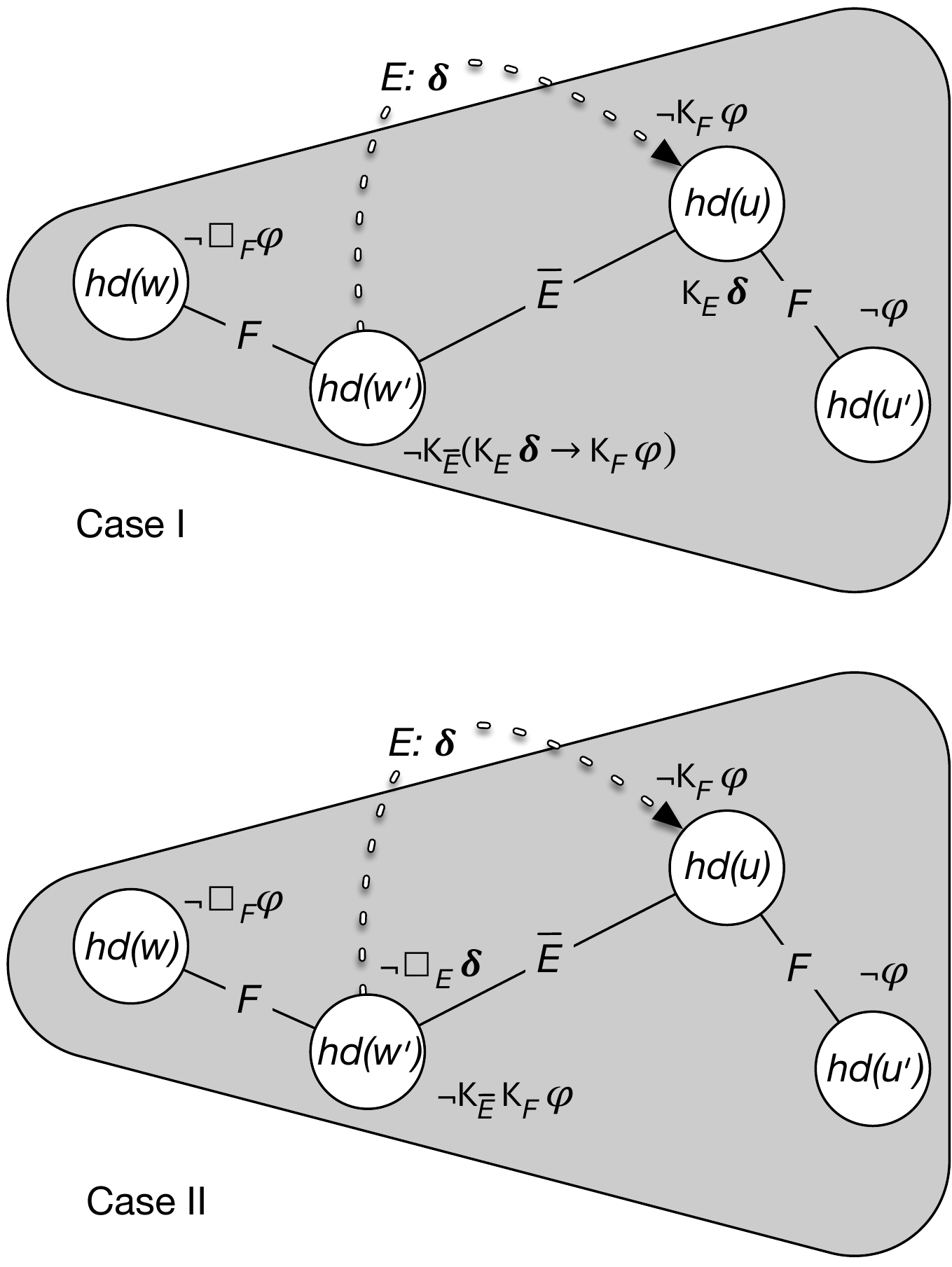}}
\caption{Towards the Proof of Lemma~\ref{H exists lemma}.}\label{pink-blue figure}
\end{center}
\end{figure}

Unfortunately, there is no way to decide upfront which of these two ways could be used to construct a consistent set $hd(w')$. Thus, in the proof below we attempt to concurrently construct both versions of the set $hd(w')$ and prove that one of the two attempts succeeds by resulting in a consistent set $hd(w)$. Finally, note that in both cases we must also guarantee that $w\sim_F w'$. To achieve this, we include in set $hd(w')$ all such formulae $\psi$ that $\K_F\psi\in hd(w)$. 

In the proof below, the two different attempts to create a set $hd(w')$ are carried out by defining sets $X$ and $Y$ and proving that at least one of them is consistent. Set $hd(w')$ is later defined as a maximal consistent extension of either set $X$ or set $Y$ depending on which one is consistent. 
\begin{proof}
Consider the following two sets of formulae:
\begin{eqnarray*}
X&=&\{\neg\K_{\overline{E}} (\K_{E}\delta\to\K_F\phi )\}\cup\{\psi\;|\;\K_F\psi\in hd(w)\},\\
Y&=&\{\neg\H_{E}\delta,\neg\K_{\overline{E}}\K_F\phi\}\cup\{\psi\;|\;\K_F\psi\in hd(w)\}.
\end{eqnarray*}
\begin{claim}
Either set $X$ or set $Y$ is consistent.
\end{claim}
\begin{proof-of-claim}
Suppose the opposite. Thus, there are 
\begin{equation}\label{phis and phi's}
    \K_F\psi_1,\dots,\K_F\psi_m,\K_F\psi'_1,\dots,\K_F\psi'_n\in hd(w)
\end{equation}
such that
\begin{eqnarray*}
\psi_{1},\dots,\psi_m&\vdash&\K_{\overline{E}} (\K_{E}\delta\to\K_F\phi ),\\
\psi'_1,\dots,\psi'_n&\vdash& \H_{E}\delta\vee\K_{\overline{E}}\K_F\phi.
\end{eqnarray*}
Then, by the Strategic Introspection axiom,
\begin{eqnarray*}
\psi_{1},\dots,\psi_m&\vdash&\K_{\overline{E}} (\K_{E}\delta\to\K_F\phi ),\\
\psi'_1,\dots,\psi'_n&\vdash& \K_E\H_{E}\delta\vee\K_{\overline{E}}\K_F\phi.
\end{eqnarray*}
Hence, by Lemma~\ref{super distributivity},
\begin{eqnarray*}
\K_F\psi_{1},\dots,\K_F\psi_m&\vdash&\K_F\K_{\overline{E}}(\K_{E}\delta\to\K_F\phi ),\\
\K_F\psi'_1,\dots,\K_F\psi'_n&\vdash& \K_F(\K_E\H_{E}\delta\vee\K_{\overline{E}}\K_F\phi).
\end{eqnarray*}
Thus, by assumption~(\ref{phis and phi's}),
\begin{eqnarray}
hd(w)&\vdash&\K_F\K_{\overline{E}} (\K_{E}\delta\to\K_F\phi )\label{KK statement},\\
hd(w)&\vdash& \K_F(\K_E\H_{E}\delta\vee\K_{\overline{E}}\K_F\phi).\nonumber
\end{eqnarray}
The last statement, by Lemma~\ref{K over vee lemma}, assumption $E\subseteq F$ of the lemma, and the Modus Ponens inference rule, implies that
$$
hd(w)\vdash \K_E\H_E\delta\vee \K_F\K_{\overline{E}}\K_F\phi.
$$
Then, by the Truth axiom and propositional reasoning,
\begin{equation}\label{vee formula}
    hd(w)\vdash \H_E\delta\vee \K_F\K_{\overline{E}}\K_F\phi.
\end{equation}
Recall that set $E$ is nonempty by the assumption of the lemma. Thus, there is at least one $e\in E$. Then, $e\in F$ by the assumption $E\subseteq F$ of the lemma. Hence, $e\in F\setminus\overline{E}$. Thus, $F\nsubseteq \overline{E}$. Then, $\vdash \K_F\K_{\overline{E}}\K_F\phi\to \H_F\phi$ by Lemma~\ref{KKK lemma}. At the same time, $hd(w)\nvdash \H_F\phi$ by the assumption $\H_F\phi\notin hd(w)$ of the lemma and the maximality of the set $hd(w)$. Then, $hd(w)\nvdash \K_F\K_{\overline{E}}\K_F\phi$ by the contraposition of the Modus Ponens inference rule. Hence, $\neg\K_F\K_{\overline{E}}\K_F\phi\in hd(w)$ because set $hd(w)$ is maximal. Thus, by propositional reasoning using statement~(\ref{vee formula}),
\begin{equation}\label{H E delta}
    hd(w)\vdash \H_E\delta.
\end{equation}
At the same time, assumption $E\subseteq F$ of the lemma implies that $(F\setminus E)\cup \overline{F}= \overline{E}$. Then, $E\cap ((F\setminus E)\cup \overline{F})=E\cap \overline{E}=\varnothing$.
Hence, the following formula
$$
\K_{E,F\setminus E}\K_{\overline{F},F\setminus E}(\K_E\delta\!\to\!\K_{E,F\setminus E}\phi)\to (\H_E\delta\!\to\!\H_{E,F\setminus E}\phi)
$$
is an instance of the Coalition-Informant-Adversary axiom
where $C=E$, $I=F\setminus E$, and $A=\overline{F}$. Thus, using statement~(\ref{H E delta}) and propositional reasoning,
$$
hd(w)\vdash \K_{E,F\setminus E}\K_{\overline{F},F\setminus E}(\K_E\delta\!\to\!\K_{E,F\setminus E}\phi)\to \H_{E,F\setminus E}\phi.
$$
Note that $E\cup(F\setminus E)=F$ and $\overline{F}\cup (F\setminus E) = \overline{E}$ by the assumption $E\subseteq F$ of the lemma. In other words,
$$
hd(w)\vdash \K_{F}\K_{\overline{E}}(\K_E\delta\to\K_{F}\phi)\to \H_{F}\phi.
$$
Then, 
$
hd(w)\vdash \H_{F}\phi
$
by statement~(\ref{KK statement}) and the Modus Ponens inference rule. Therefore, $\H_F\phi\in hd(w)$ because set $hd(w)$ is maximal, which contradicts assumption $\H_F\phi\notin hd(w)$ of the lemma.
\end{proof-of-claim}

The claim that we just proved states that either set $X$ or set $Y$ is consistent. We consider these two cases separately.

\vspace{1mm}
\noindent{\bf Case I:} set $X$ is consistent. Let $X'$ be any maximal consistent extension of the set $X$ and let state $w'$ be the sequence $w::F::X'$. Note that $w\in W$ by Definition~\ref{canonical state} and the choice of set $X$, set $X'$, and sequence $w'$. Also, $w\sim_F w'$ by Definition~\ref{canonical sim} and the choice of sequence $w'$.

Note that $\neg\K_{\overline{E}} (\K_{E}\delta\to\K_F\phi )\in X\subseteq X'=hd(w')$ by the choice of set $X$, set $X'$, and sequence $w'$. Thus, $\K_{\overline{E}} (\K_{E}\delta\to\K_F\phi )\notin hd(w')$ because set $hd(w')$ is consistent. Hence, by Lemma~\ref{K exists lemma}, there is a state $u\in W$ such that $w'\sim_{\overline{E}} u$ and $\K_{E}\delta\to\K_F\phi\notin hd(u)$. Then, $\K_{E}\delta\in hd(u)$ and $\K_F\phi\notin hd(u)$ because $hd(u)$ is a maximal consistent set. 
Statements $w'\sim_{\overline{E}} u$ and $\K_{E}\delta\in hd(u)$ imply that $(w',E,\delta,u)\in M$ by Definition~\ref{canonical mechanism}.
Statement $\K_F\phi\notin hd(u)$ implies that there is a state $u'\in W$ such that $u\sim_F u'$ and $\phi\notin hd(u')$ by Lemma~\ref{K exists lemma}.

\vspace{1mm}
\noindent{\bf Case II:} set $Y$ is consistent. Let $Y'$ be any maximal consistent extension of the set $Y$ and let state $w'$ be the sequence $w::F::X'$. As in the previous case, $w\in W$ by Definition~\ref{canonical state} and the choice of set $Y$, set $Y'$, and sequence $w'$. Also, $w\sim_F w'$ by Definition~\ref{canonical sim} and the choice of $w'$.

Note that $\neg\K_{\overline{E}}\K_F\phi \in Y \subseteq Y' = hd(w')$ by the choice of set $Y$, set $Y'$, and sequence $w'$. Thus, $\K_{\overline{E}}\K_F\phi \notin hd(w')$ as set $hd(w')$ is maximal consistent. Hence, by Lemma~\ref{K exists lemma}, there is a state $u\in W$ such that 
\begin{equation}\label{two facts}
    w'\sim_{\overline{E}} u \;\;\;\mbox{ and }\;\;\;\K_F\phi\notin hd(u).
\end{equation}

At the same time, $\neg\H_E\delta\in Y\subseteq Y'=hd(w')$ by the choice of set $Y$, set $Y'$, and sequence $w'$. Thus, $\H_E\delta\notin hd(w')$ because set $hd(w')$ is consistent. Then, $(w',E,\delta,u)\in M$ by Definition~\ref{canonical mechanism} and because $w'\sim_{\overline{E}} u$ by statement~(\ref{two facts}).

Finally, $K_F\phi \notin hd(u)$ by statement~(\ref{two facts}). Therefore, by Lemma~\ref{K exists lemma}, there exists a state $u' \in W$ such that $u\sim_F u'$ and $\phi\notin hd(u')$.
This concludes the proof of the lemma.
\end{proof}

The next ``truth lemma'' follows from the four previous lemmas in the standard way. Due to the space constraint, we give its proof in the appendix.

\begin{lemma}\label{induction lemma}
$w\Vdash \phi$ iff $\phi\in hd(w)$.
\end{lemma}

\noindent{\bf Theorem \ref{completeness}.} {\em 
If $X\nvdash\phi$, then there is  a state $w$ of a clandestine  game such that $w\Vdash\chi$ for each formula $\chi\in X$ and $w\nVdash\phi$.
}

\begin{proof}
If $X\nvdash\phi$, then set $X\cup\{\neg\phi\}$ is consistent. Let $w$ be any maximal consistent extension of this set. Then, $w\Vdash\chi$ for each formula $\chi\in X$ and $w\vdash \neg\phi$ by Lemma~\ref{induction lemma}. Therefore, $w\nvdash \phi$ by item 2 of Definition~\ref{sat}.
\end{proof}

\section{Conclusion}

In this paper, we proposed a sound and complete logical system that describe properties of clandestine power modality $\Box_C\phi$. A natural generalization of our work could be a study of ``partially-clandestine'' modality $\Box^F_C\phi$, that stands for ``coalition $C$ knows an operation that it can use to achieve $\phi$ unnoticeable to anyone outside (friendly) coalition $F$''. 

It is also possible to consider a broader class of clandestine operations that achieve a goal through several consecutive clandestine actions of the given coalition. This type of multi-step operations is similar to multi-step strategies studied in know-how logics~\cite{fhlw17ijcai,lw17icla,w17synthese,w15lori}. 

\bibliographystyle{named}
\bibliography{sp}


\clearpage

\appendix

\begin{center}\Large\sc
    Technical Appendix
\end{center}

\section{Soundness}

In this section, we prove the soundness of our logical system. The soundness of the Truth, the Negative Introspection, the Distributivity, and the Monotonicity axioms as well as of the Modus Ponens inference rule and the Necessitation inference rule for modality $\K$ are well-known results about epistemic logic of distributed knowledge~\cite{fhmv95}. Below we prove the soundness of each of the remaining axioms and the  Necessitation inference rule for modality $\Box$ as separate lemmas.

\begin{lemma}
If $w\Vdash \Box_C\phi$, then $w\Vdash \K_C\Box_C\phi$.
\end{lemma}
\begin{proof}
Consider any state $v$ such that $w\sim_C v$. By item 4 of Definition~\ref{sat}, it suffices to show that $v\Vdash \Box_C\phi$.

By item 5 of Definition~\ref{sat}, assumption $w\Vdash \Box_C\phi$ implies that there is a nonempty coalition $C'\subseteq C$ and an operation $\delta\in\Delta$ such that for any  states $w',u,u'\in W$, 
\begin{equation}\label{aug27-long}
    w\sim_C w'\wedge (w',C',\delta,u)\in M\wedge u\sim_C u'\Rightarrow u'\Vdash \phi.
\end{equation}

Towards the proof of statement $v\Vdash \Box_C\phi$, consider any states $v',z,z'\in W$ such that $v\sim_C v'$, $(v',C',\delta,z)\in M$, and $z\sim_C z'$. By item 5 of Definition~\ref{sat}, it suffices to show that $z'\Vdash\phi$. Indeed, notice that $w\sim_C v'$ because $w\sim_C v$ and $v\sim_C v'$. Therefore, $z'\Vdash\phi$ due to statement~(\ref{aug27-long}).
\end{proof}

\begin{lemma}
If $w\Vdash \K_{C,I}\K_{A,I}(\K_C\phi\to\K_{C,I}\psi)$, $w\Vdash \H_C\phi$, and $C\cap (I\cup A)=\varnothing$, then $w\Vdash \H_{C,I}\psi$.
\end{lemma}
\begin{proof}
Assumption $w\Vdash \H_C\phi$ implies that there is a coalition
\begin{equation}\label{aug30-C'}
    C'\subseteq C
\end{equation}
and an action $\delta\in \Delta$ such that
for any  states $w',u,u'\in W$, 
\begin{equation}\label{aug30-long}
    w\sim_C w'\wedge (w',C',\delta,u)\in M\wedge u\sim_C u'\Rightarrow u'\Vdash \phi.
\end{equation}
Towards the proof of statement $w\Vdash \Box_{C,I}\psi$, consider any states $w'',v,v'\in W$ such that 
\begin{equation}\label{aug30-ww''}
    w\sim_{C,I} w'',
\end{equation}
\begin{equation}\label{aug30-M}
    (w'',C',\delta,v)\in M,
\end{equation}
and 
\begin{equation}\label{aug30-vv'}
    v\sim_{C,I} v'.
\end{equation}
By item 5 of Definition~\ref{sat}, it suffices to show that $v'\Vdash\psi$. 


\begin{claim}
$v\Vdash\K_C\phi$.
\end{claim}
\begin{proof-of-claim}
Consider any state $v''\in W$ such that
\begin{equation}\label{aug30-vv''}
    v\sim_C v''.
\end{equation}
By item 4 of Definition~\ref{sat}, it suffices to show that 
$v''\Vdash\phi$. Indeed, statement~(\ref{aug30-ww''}) implies that $w\sim_{C} w''$. Then, $v''\Vdash\phi$ by statement~(\ref{aug30-long}), statement~(\ref{aug30-M}), and assumption~(\ref{aug30-vv''}).
\end{proof-of-claim}

By item 4(a) of Definition~\ref{game}, the assumption~(\ref{aug30-M}) implies that $w''\sim_{\overline{C'}} v$. Thus, $w''\sim_{\overline{C}} v$ because $\overline{C}\subseteq\overline{C'}$, see assumption~(\ref{aug30-C'}). Hence, by the assumption $C\cap (I\cup A)=\varnothing$ of the lemma,
\begin{equation}\label{aug30-w''v}
    w''\sim_{A,I} v.
\end{equation}

At the same time, by item 4 of Definition~\ref{sat}, assumption  $w\Vdash \K_{C,I}\K_{A,I}(\K_C\phi\to\K_{C,I}\psi)$ of the lemma and statement~(\ref{aug30-ww''}) imply that
$$
w''\Vdash \K_{A,I}(\K_C\phi\to\K_{C,I}\psi).
$$
Then, by statement~(\ref{aug30-w''v}) and item 4 of Definition~\ref{sat},
$$
v\Vdash \K_C\phi\to\K_{C,I}\psi.
$$
Thus, 
$
v\Vdash \K_{C,I}\psi
$
by the earlier {\em Claim}, and part 3 of Definition~\ref{sat}. Therefore,
$
v'\Vdash \psi
$
by statement~(\ref{aug30-vv'}) and part 5 of Definition~\ref{sat}.
\end{proof}

\begin{lemma}
$w\nVdash \Box_C\bot$.
\end{lemma}
\begin{proof}
Suppose that $w\Vdash \Box_C\bot$. Thus, by item 5 of Definition~\ref{sat}, there is a nonempty coalition $C'\subseteq C$ and  an operation $\delta\in \Delta$ such that for any states $w',u,u'\in W$, if $w\sim_C w'$, $(w',C',\delta,u)\in M$, and $u\sim_C u'$, then $u'\Vdash \phi$. In particular, 
\begin{equation}\label{aug27-short}
    \forall u\in W( (w,C',\delta,u)\in M \Rightarrow u\Vdash \bot).
\end{equation}
By item 4(b) of Definition~\ref{game}, there is a state $v\in W$ such that $(w,C',\delta,v)\in M$. Therefore, $v\Vdash\bot$ by statement~(\ref{aug27-short}), which is a contradiction.
\end{proof}

\begin{lemma}
$w\nVdash \Box_\varnothing\phi$.
\end{lemma}
\begin{proof}
The statement of the lemma holds by item 5 of Definition~\ref{sat} because the empty set has no nonempty subsets.
\end{proof}

\begin{lemma}
If $w\Vdash \phi$ for each state $w$ of each clandestine game and coalition $C$ is nonempty, then $w\Vdash\Box_C\phi$ for each state $w$ of each clandestine game.
\end{lemma}
\begin{proof}
Consider any state $w$ of a clandestine game. By item 3 of Definition~\ref{game}, set $\Delta$, contains at least one element~$\delta$.

Towards the proof of $w\Vdash\Box_C\phi$, consider  any states $w',u,u'\in W$ such that $w\sim_C w'$, $(w',C,\delta,u)\in M$, and $u\sim_C u'$. By item 5 of Definition~\ref{sat}, because set $C$ is nonempty, it suffices to show that $u'\Vdash \phi$. The latter is true by the assumption of the lemma.
\end{proof}

\section{Auxiliary Lemmas}

In this section of the appendix, we prove 
Lemma~\ref{super distributivity},
Lemma~\ref{positive introspection lemma},
Lemma~\ref{KKK lemma},
Lemma~\ref{K neg K to neg H lemma}, and 
Lemma~\ref{K over vee lemma}
stated in the main part of the paper.

\begin{lemma}[deduction]\label{deduction lemma}
If $X,\phi\vdash\psi$, then $X\vdash\phi\to\psi$.
\end{lemma}
\begin{proof}
Suppose that sequence $\psi_1,\dots,\psi_n$ is a proof from set $X\cup\{\phi\}$ and the theorems of our logical system that uses the Modus Ponens inference rule only. In other words, for each $k\le n$, either
\begin{enumerate}
    \item $\vdash\psi_k$, or
    \item $\psi_k\in X$, or
    \item $\psi_k$ is equal to $\phi$, or
    \item there are integers $i,j<k$ such that formula $\psi_j$ is equal to $\psi_i\to\psi_k$.
\end{enumerate}
It suffices to show that $X\vdash\phi\to\psi_k$ for each $k\le n$. We prove this by induction on $k$ through considering the four cases above separately.

\vspace{1mm}
\noindent{\bf Case I}: $\vdash\psi_k$. Note that $\psi_k\to(\phi\to\psi_k)$ is a propositional tautology, and thus, is an axiom of our logical system. Hence, $\vdash\phi\to\psi_k$ by the Modus Ponens inference rule. Therefore, $X\vdash\phi\to\psi_k$. 

\vspace{1mm}
\noindent{\bf Case II}: $\psi_k\in X$. Then, similar to the previous case, $X\vdash\phi\to\psi_k$.

\vspace{1mm}
\noindent{\bf Case III}: formula $\psi_k$ is equal to $\phi$. Thus, $\phi\to\psi_k$ is a propositional tautology. Therefore, $X\vdash\phi\to\psi_k$. 

\vspace{1mm}
\noindent{\bf Case IV}:  formula $\psi_j$ is equal to $\psi_i\to\psi_k$ for some $i,j<k$. Thus, by the induction hypothesis, $X\vdash\phi\to\psi_i$ and $X\vdash\phi\to(\psi_i\to\psi_k)$. Note that formula 
$$
(\phi\to\psi_i)\to((\phi\to(\psi_i\to\psi_k))\to(\phi\to\psi_k))
$$
is a propositional tautology. Therefore, $X\vdash \phi\to\psi_k$ by applying the Modus Ponens inference rule twice.
\end{proof}

\renewcommand*{\thelemma}{\ref{super distributivity}}
\begin{lemma}
If $\phi_1,\!\dots,\!\phi_n\vdash\psi$, then $\K_a\phi_1,\!\dots,\!\K_a\phi_n\vdash\K_a\psi$.
\end{lemma}
\begin{proof}
By Lemma~\ref{deduction lemma} applied $n$ times, the assumption $\phi_1,\dots,\phi_n\vdash\psi$ implies that
$$
\vdash\phi_1\to(\phi_2\to\dots(\phi_n\to\psi)\dots).
$$
Thus, 
$$
\vdash\K_a(\phi_1\to(\phi_2\to\dots(\phi_n\to\psi)\dots))
$$
by the Necessitation inference rule.
Hence, by the Distributivity axiom and the Modus Ponens inference rule,
$$
\vdash\K_a\phi_1\to\K_a(\phi_2\to\dots(\phi_n\to\psi)\dots).
$$
Then, 
$
\K_a\phi_1\vdash\K_a(\phi_2\to\dots(\phi_n\to\psi)\dots)
$,
again by the Modus Ponens inference rule.
Therefore, $\K_a\phi_1,\dots,\K_a\phi_n\vdash\K_a\psi$ by applying the previous steps $(n-1)$ more times.
\end{proof}

\renewcommand*{\thelemma}{\ref{positive introspection lemma}}
\begin{lemma}[positive introspection]
$\vdash \K_a\phi\to\K_a\K_a\phi$. 
\end{lemma}
\begin{proof}
Formula $\K_a\neg\K_a\phi\to\neg\K_a\phi$ is an instance of the Truth axiom. Thus, by contraposition, $\vdash \K_a\phi\to\neg\K_a\neg\K_a\phi$. Hence, taking into account the following instance of  the Negative Introspection axiom: $\neg\K_a\neg\K_a\phi\to\K_a\neg\K_a\neg\K_a\phi$,
\begin{equation}\label{pos intro eq 2}
\vdash \K_a\phi\to\K_a\neg\K_a\neg\K_a\phi.
\end{equation}

At the same time, $\neg\K_a\phi\to\K_a\neg\K_a\phi$ is an instance of the Negative Introspection axiom. Thus, $\vdash \neg\K_a\neg\K_a\phi\to \K_a\phi$ by contraposition. Hence, by the Necessitation inference rule, 
$\vdash \K_a(\neg\K_a\neg\K_a\phi\to \K_a\phi)$. Thus, by  the Distributivity axiom and the Modus Ponens inference rule, 
$
  \vdash \K_a\neg\K_a\neg\K_a\phi\to \K_a\K_a\phi.
$
The latter, together with statement~(\ref{pos intro eq 2}), implies the statement of the lemma by propositional reasoning.
\end{proof}

\noindent{\bf Lemma~\ref{KKK lemma}}
{\em $\vdash \K_F\K_{E}\K_F\phi\to \H_F\phi$, where $F\nsubseteq E$.}

\begin{proof}
Note that sets $F\setminus E$, $F\cap E$, and $E\setminus F$ are pairwise disjoint for any sets $E$ and $F$. Then, by the Coalition-Informant-Adversary axiom, where $C=F\setminus E$, $I=F\cap E$, and $A=E\setminus F$,
\begin{eqnarray*}
&&\vdash\K_{F\setminus E,F\cap E}\K_{E\setminus F,F\cap E}(\K_{F\setminus E}\top\to\K_{F\setminus E,F\cap E}\phi)\\
&&\hspace{5mm}\to(\H_{F\setminus E}\top\to\H_{F\setminus E,F\cap E}\phi).
\end{eqnarray*}
In other words, 
\begin{eqnarray}
\vdash\K_{F}\K_{E}(\K_{F\setminus E}\top\to\K_{F}\phi)
\to(\H_{F\setminus E}\top\to\H_{F}\phi).\label{hold me}
\end{eqnarray}
At the same time, formula
$
\K_{F}\phi\to
(\K_{F\setminus E}\top\to\K_{F}\phi)
$
is a propositional tautology. Thus, by the Necessitation inference rule,
$
\vdash\K_{E}(\K_{F}\phi\to
(\K_{F\setminus E}\top\to\K_{F}\phi))
$.
Then, by the Distributivity axiom and the Modus Ponens inference rule,
$
\vdash\K_{E}\K_{F}\phi\to
\K_{E}(\K_{F\setminus E}\top\to\K_{F}\phi)
$.
Hence, again by the Necessitation inference rule,
$$
\vdash\K_{F}(\K_{E}\K_{F}\phi\to
\K_{E}(\K_{F\setminus E}\top\to\K_{F}\phi)).
$$
Thus, by the Distributivity axiom and the Modus Ponens inference rule,
$
\vdash \K_{F}\K_{E}\K_{F}\phi\to
\K_{F}\K_{E}(\K_{F\setminus E}\top\to\K_{F}\phi).
$
Then, 
$
\vdash\K_{F}\K_{E}\K_{F}\phi\to(\H_{F\setminus E}\top\to\H_{F}\phi)
$
by the laws of propositional reasoning using statement~(\ref{hold me}).
Finally, note that assumption $F\nsubseteq E$ implies that set $F\setminus E$ is not empty. Hence,  $\vdash \H_{F\setminus E}\top$ by the Necessitation inference rule. Therefore,
$
\vdash\K_{F}\K_{E}\K_{F}\phi\to\H_{F}\phi
$
by the laws of propositional reasoning.
\end{proof}

\noindent{\bf Lemma~\ref{K neg K to neg H lemma}}
{\em $\vdash \K_E\neg\K_F\phi\to \neg\H_F\phi$, where $E\cap F=\varnothing$.
}
\begin{proof}
Assumption $E\cap F=\varnothing$ implies that set $F$ and set $E\cup\varnothing$ are disjoint. Then, by the Coalition-Informant-Adversary axiom, where
$C= F$, $I = E$, $A=\varnothing$, and $\psi=\bot$,
\begin{equation}\label{hold 2}
  \vdash \K_{F,E}\K_{E}(\K_F\phi\to\K_{F,E}\bot)\to(\H_F\phi\to\H_{F,E}\bot).  
\end{equation}
At the same time, the formula
$
\neg\K_F\phi\to (\K_F\phi\to\K_{F,E}\bot)
$
is a tautology. Thus, 
$
\vdash \K_E(\neg\K_F\phi\to (\K_F\phi\to\K_{F,E}\bot))
$
by the Necessitation inference rule. Hence, by the Distributivity axiom and the Modus Ponens inference rule,
$$
\vdash \K_E\neg\K_F\phi\to \K_E(\K_F\phi\to\K_{F,E}\bot).
$$
Then, 
$
\vdash \K_E\neg\K_F\phi\to \K_E\K_E(\K_F\phi\to\K_{F,E}\bot)
$
using propositional reasoning and Lemma~\ref{positive introspection lemma}.
Thus, using the Monotonicity axiom and propositional reasoning,
$$
\vdash \K_E\neg\K_F\phi\to \K_{F,E}\K_E(\K_F\phi\to\K_{F,E}\bot).
$$
Hence, 
$
\vdash \K_E\neg\K_F\phi\to(\H_F\phi\to\H_{F,E}\bot) 
$
using statement~(\ref{hold 2}) and propositional reasoning.
Finally, note that $\neg\H_{F,E}\bot$ is an instance of the Nontermination axiom. Therefore,
$
\vdash \K_E\neg\K_F\phi\to\neg\H_F\phi  
$
by the laws of propositional reasoning.
\end{proof}

\noindent{\bf Lemma~\ref{K over vee lemma}}
{\em 
$\vdash\K_F(\K_E\phi\vee\psi)\to\K_E\phi\vee \K_F\psi$, where $E\subseteq F$.
}


\begin{proof}
Note that $\K_E\phi\vee\psi\to(\neg\K_E\phi\to\psi)$ is a propositional tautology. Thus, $\vdash\K_F(\K_E\phi\vee\psi\to(\neg\K_E\phi\to\psi))$ by the Necessitation inference rule. Hence, by the Distributivity axiom and the Modus Ponens inference rule,
$$
\vdash\K_F(\K_E\phi\vee\psi)\to\K_F(\neg\K_E\phi\to\psi).
$$
Then, 
$
    \K_F(\K_E\phi\vee\psi)\vdash \K_F(\neg\K_E\phi\to\psi)
$
by the Modus Ponens inference rule.
Thus, by the Distributivity axiom and the Modus Ponens,
$
    \K_F(\K_E\phi\vee\psi)\vdash \K_F\neg\K_E\phi\to\K_F\psi.
$
Hence,
$
    \K_F(\K_E\phi\vee\psi)\vdash \K_E\neg\K_E\phi\to\K_F\psi
$
by the Monotonicity axiom, assumption $E\subseteq F$, and the laws of propositional reasoning.
Then, by the Negative Introspection axiom and propositional reasoning,
$
    \K_F(\K_E\phi\vee\psi)\vdash \neg\K_E\phi\to\K_F\psi.
$
Thus, by  propositional reasoning,
$
    \K_F(\K_E\phi\vee\psi)\vdash \K_E\phi\vee\K_F\psi.
$
Therefore, $\vdash\K_F(\K_E\phi\vee\psi)\to\K_E\phi\vee \K_F\psi$ by the deduction theorem.
\end{proof}

\section{Completeness}


\noindent{\bf Lemma~\ref{induction lemma}}
{\em 
$w\Vdash \phi$ iff $\phi\in hd(w)$.
}
\vspace{0mm}

\begin{proof}
We prove the statement of the lemma by structural induction on formula $\phi$. If $\phi$ is a propositional variable, then the statement of the lemma follows from item 1 of Definition~\ref{sat} and Definition~\ref{canonical pi}. If formula $\phi$ is a negation or an implication, then the required follows from items 2 and 3 of Definition~\ref{sat} and the maximality and the consistency of the set $hd(w)$ in the standard way.

Suppose that formula $\phi$ has the form $\K_C\psi$.

\noindent$(\Rightarrow)$: If $\K_C\psi\notin hd(w)$, then, by Lemma~\ref{K exists lemma}, there  is a state $u\in W$ such that $w\sim_C u$ and $\psi\notin hd(u)$. Thus, $u\nVdash\psi$ by the induction hypothesis. Therefore, $w\nVdash\K_C\psi$ by item 4 of Definition~\ref{sat}.

\noindent$(\Leftarrow)$: If $\K_C\psi\in hd(w)$, then $\psi\in hd(u)$ for each state $u\in W$ such that $w\sim_C u$ by Lemma~\ref{K all lemma}. Hence, by the induction hypothesis, $u\Vdash \psi$ for each state $u\in W$ such that $w\sim_C u$. Therefore, $w\Vdash\K_C\psi$ by item 4 of Definition~\ref{sat}.

Assume that formula $\phi$ has the form $\H_C\psi$.

\noindent$(\Rightarrow)$: If $\H_C\psi\notin hd(w)$, then, by Lemma~\ref{H exists lemma}, for any nonempty coalition $C'\subseteq C$ and any action $\delta\in \Delta$, there  are states $w',u,u'$ such that $w\sim_{C} w'$, $(w',C',\delta,u)\in M$, $u\sim_C u'$, and $\psi\notin hd(u')$. Thus, by the induction hypothesis, for any nonempty coalition $C'\subseteq C$ and any action $\delta\in \Delta$, there  are states $w',u,u'$ such that $w\sim_{C} w'$, $(w',C',\delta,u)\in M$, $u\sim_C u'$, and $u'\nVdash\psi$. Therefore, $w\nVdash\H_C\psi$ by item 5 of Definition~\ref{sat}.

\noindent$(\Leftarrow)$: Let $\H_C\psi\in hd(w)$. Thus, by Lemma~\ref{H all lemma},  for any three states $w',u,u'\in W$, if $w\sim_C w'$, $(w',C,\psi,u)\in M$, and $u\sim_C u'$, then $\phi\in hd(u')$. Hence, by the induction hypothesis,  for any three states $w',u,u'\in W$, if $w\sim_C w'$, $(w',C,\psi,u)\in M$, and $u\sim_C u'$, then $u'\Vdash \psi$. Therefore, $w\Vdash\H_C\psi$ by item 5 of Definition~\ref{sat}.
\end{proof}

\end{document}